\documentclass{llncs}

\usepackage{graphicx,color,xspace}
\usepackage{comment}
\usepackage{amsmath,amssymb}
\graphicspath{{images/}}

\usepackage{fullpage}

\newcommand{\Ex}{\mathbb{E}}

\newcommand{\real}{\mathbb{R}}
\newcommand{\etal}{\textit{et al.}}

\newcommand{\namedref}[2]{\hyperref[#2]{#1~\ref*{#2}}}

\title{Balanced Line Separators of Unit Disk Graphs\thanks{%
A preliminary version was presented at 15th Algorithms and Data Structures Symposium (WADS 2017).
Chiu, van Renssen and Roeloffzen were supported by JST ERATO Grant Number JPMJER1201, Japan.
Chiu was also supported by ERC STG 757609. 
Korman was partially supported by MEXT KAKENHI No.~17K12635 and the NSF award CCF-1422311.
Katz was partially supported by grant 1884/16 from the Israel Science Foundation.
Okamoto was partially supported by KAKENHI Grant Numbers JP24106005, JP24220003 and JP15K00009, JST CREST Grant Number JPMJCR1402, and Kayamori Foundation for Informational Science Advancement.
Smorodinsky's research was partially supported by Grant 635/16 from the Israel Science Foundation.%
}}

\author{%
Paz Carmi\inst{1}
\and
Man Kwun Chiu\inst{2}
\and
Matthew J. Katz\inst{1}
\and
Matias Korman\inst{3}%
\and
Yoshio Okamoto\inst{4,5}%
\and\\
Andr\'e van Renssen\inst{6}
\and
Marcel Roeloffzen\inst{7}
\and
Taichi Shiitada\inst{4}
\and
Shakhar Smorodinsky\inst{1}%
}

\institute{%
Ben-Gurion University of the Negev, Beer-Sheva, Israel
\and
Institut f\"ur Informatik, Freie Universit\"at Berlin, Berlin, Germany
\and
Tufts University, Medford, MA, USA
\and
The University of Electro-Communications, Tokyo, Japan
\and
RIKEN Center for Advanced Intelligence Project, Tokyo, Japan
\and
University of Sydney, Sydney, Australia
\and
TU Eindhoven, Eindhoven, the Netherlands
}

\begin{document}

\maketitle

\begin{abstract}
We prove a geometric version of the graph separator theorem for the unit disk intersection graph: for any set of $n$ unit disks in the plane there exists a line $\ell$ such that $\ell$ intersects at most $O(\sqrt{(m+n)\log{n}})$ disks and each of the halfplanes determined by $\ell$ contains at most $2n/3$ unit disks from the set, where $m$ is the number of intersecting pairs of disks. We also show that an axis-parallel line intersecting $O(\sqrt{m+n})$ disks exists, but each halfplane may contain up to $4n/5$ disks. We give an almost tight lower bound (up to sublogarithmic factors) for our approach, and also show that no line-separator of sublinear size in $n$ exists when we look at disks of arbitrary radii, even when $m=0$. 
Proofs are constructive and suggest simple algorithms that run in linear time.
Experimental evaluation has also been conducted, which shows that for random instances our method outperforms the method by Fox and Pach (whose separator has size $O(\sqrt{m})$).
\end{abstract}

\section{Introduction}

Balanced separators in graphs are a fundamental tool and used in many divide-and-conquer-type algorithms as well as for proving theorems by induction.
Given an undirected graph $G=(V,E)$ with $V$ as its vertex set and $E$ as its edge set, and a non-negative real number $\alpha \in [1/2,1]$, we say that a subset $S\subseteq V$ is an \emph{$\alpha$-separator}
if the vertex set of $G \setminus S$ can be partitioned into two sets $A$ and $B$, each of size at most $\alpha |V|$ such that there is no edge between $A$ and $B$.
The parameter $\alpha$ determines how balanced the two sets $A$ and $B$ are in terms of size. For a balanced separator to be useful we want both the size $|S|$ of the separator and $\alpha \geq 1/2$ to be small.

Much work has been done to prove the existence of separators with certain properties in general sparse graphs.
For example, the well-known Lipton--Tarjan planar separator theorem~\cite{doi:10.1137/0136016}
states that for any $n$-vertex planar 
graph, there exists a $2/3$-separator of size $O(\sqrt{n})$.
Similar theorems have been proven for 
bounded-genus graphs~\cite{DBLP:journals/jal/GilbertHT84},
minor-free graphs~\cite{AST}, 
low-density graphs, and graphs with polynomial expansion~\cite{DBLP:journals/ejc/NesetrilM08a,DBLP:journals/siamcomp/Har-PeledQ17}.

These separator results apply to graph classes that do not contain complete graphs of arbitrary size, and each graph in the classes contains only $O(n)$ edges,
where $n$ is the number of vertices.
Since any $\alpha$-separator of a complete graph has $\Omega(n)$ vertices, the study of separators for graph classes that contain complete graphs seems useless.
However, it is not clear how small a separator can be with respect to the number of edges for possibly dense graphs.

Our focus of interest is possibly dense geometric graphs, which often encode additional geometric information other than adjacency. Even though one can use the separator tools in geometric graphs, often the geometric information is lost in the process. As such, a portion of the literature has focused on the search of balanced separators that also preserve the geometric properties of the geometric graph. Such separators are called \emph{geometric separators}. 

Among several others, we highlight the work of
Miller~\etal~\cite{DBLP:journals/jacm/MillerTTV97}, and
Smith and Wormald~\cite{DBLP:conf/focs/SmithW98}. They considered intersection 
graphs of $n$ balls in $\mathbb{R}^d$
and proved that if every point in $d$-dimensional space is covered by at most $k$ of the given balls,
then there exists a $(d+1)/(d+2)$-separator of size $O(k^{1/d}n^{1-1/d})$ (and such a separator can be found in deterministic linear time~\cite{DBLP:journals/fuin/EppsteinMT95}). More interestingly, the separator itself and the two sets it creates have very nice properties; they show that there exists a $(d{-}1)$-dimensional sphere that intersects at most $O(k^{1/d}n^{1-1/d})$ balls and contains at most $(d+1)n/(d+2)$ balls in its interior and at most $(d+1)n/(d+2)$ balls in its exterior. In this case, the sphere acts as the separator (properly speaking, the balls that intersect the sphere), whereas the two sets $A$ and $B$ are the balls that are inside and outside the separator sphere, respectively. Note that the graph induced by the set $A$ consists of the intersection graph of the balls inside the separator (similarly, $B$ for the balls outside the separator and $S$ for the balls intersecting the sphere).

We emphasize that, even though the size of the separator is larger than the one from Lipton--Tarjan for planar graphs (specially for high values of $d$), the main advantage is that the three subgraphs it creates are geometric graphs of the same family (intersection graphs of balls in $\mathbb{R}^d$).
The bound on the separator size does not hold up well when $k$ is large, even for $d=2$: if $\sqrt{n}$ disks overlap at a single point and the other disks form a path we have $k=\sqrt{n}$ and $m=\Theta(n)$, where $m$ is the number of edges in the intersection graph. Hence, the separator has size $O(\sqrt{kn})=O(m^{3/4})$.

Fox and Pach~\cite{Fox20081070} gave another separator result that follows the same spirit: the intersection graph
of a set of Jordan curves in the plane
has a $2/3$-separator of size $O(\sqrt{m})$
if every pair of curves intersects at a constant number of 
points.\footnote{Without restriction on the number of
intersection points for every pair of curves, the bound of
$O(\sqrt{m}\log m)$ can be achieved \cite{DBLP:journals/cpc/Matousek14}.
} A set of disks in $\mathbb{R}^2$ satisfies this condition, and thus the theorem applies to disk graphs. Their proof can be turned into a polynomial-time algorithm.
However, we need to construct the arrangement of disks,
which takes $O(n^2 2^{\alpha(n)})$ time, where 
$\alpha(n)$ is
the inverse Ackermann function~\cite{DBLP:journals/tcs/EdelsbrunnerGPPSS92},
and in practice an efficient implementation is non-trivial.

From a geometric perspective these two results show that, given a set of unit disks in the plane, we can always find a closed curve in the plane (a circle~\cite{DBLP:journals/jacm/MillerTTV97,DBLP:conf/focs/SmithW98} and a Jordan curve~\cite{Fox20081070}, respectively) to partition the set. The disks intersected by the curve are those in the separator, and the two disjoint sets are the disks inside and outside the curve, respectively. 

\paragraph{Results and Paper Organization.} 
In this paper we continue the idea of geometric separators and show that a balanced separator always exists, even if we constrain the separator to be a line (see \figurename~\ref{fig:linesep}). Given a set of $n$ unit disks with $m$ pairs of intersecting disks, we show that a line $2/3$-separator of size $O(\sqrt{(m+n) \log n})$ can be found in expected $O(n)$ time, and that an axis-parallel line $4/5$-separator of size $O(\sqrt{m+n})$
can be found in deterministic $O(n)$ time.

\begin{figure}[t]
  \centering
  \includegraphics{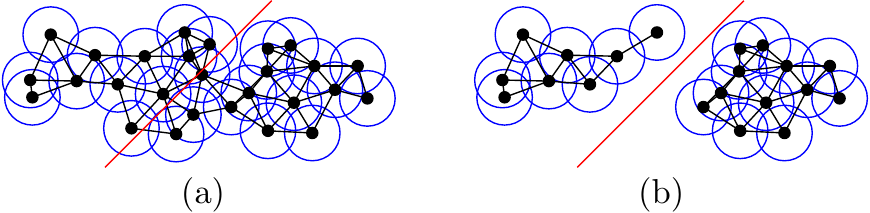}
  \caption{An example of a line separator of a unit disk graph.
    (a) A family of unit disks (blue) and a line (red).
    (b) Removing the disks intersected by the red line leaves a disconnected
    graph.
  }
  \label{fig:linesep}
\end{figure}

Comparing our results with the previous work, our algorithm matches or improves in four ways, see also \tablename~\ref{tab:results-comp}. $(i)$ simplicity of the shape: circle~\cite{DBLP:journals/jacm/MillerTTV97,DBLP:conf/focs/SmithW98} vs.\ Jordan curve~\cite{Fox20081070} vs.\ our line, $(ii)$ balance of the sets $A$ and $B$: $3/4$~\cite{DBLP:journals/jacm/MillerTTV97,DBLP:conf/focs/SmithW98} vs.\ $2/3$ for both~\cite{Fox20081070} and us, $(iii)$ size of the separator: $O(m^{3/4})$~\cite{DBLP:journals/jacm/MillerTTV97,DBLP:conf/focs/SmithW98} vs.\ $O(\sqrt{m})$~\cite{Fox20081070} vs.\ our $\tilde{O}(\sqrt{m})$.\footnote{The $\tilde{O}(\cdot)$ notation suppresses sublogarithmic factors. In particular, we note that our separator is slightly larger than the Fox-Pach separator.} Finally, $(iv)$ our algorithms are simple and asymptotically faster: $O(n)$~\cite{DBLP:journals/jacm/MillerTTV97,DBLP:conf/focs/SmithW98} vs.\ $\tilde{O}(n^2)$~\cite{Fox20081070} vs.\ our $O(n)$. 
Note that those algorithms require \emph{geometric representations}.
For example, unit disk graphs are given by a set of unit disks, not as a graph
combinatorially.
Indeed, finding a geometric representaion of a unit disk graph is not a
trivial task:
this problem is NP-hard~\cite{DBLP:journals/comgeo/BreuK98}, 
and $\exists$$\mathbb{R}$-complete~\cite{DBLP:journals/dcg/KangM12}.

\begin{table}[t]
\centering
\caption{Comparison of our results with other geometric separator results.}
\begin{tabular}{c | c | c | c | c |c }
\textbf{result} & \textbf{separator shape} & \textbf{balance} & \textbf{separator size} & \textbf{run-time} & \textbf{object type}\\
\hline 
\cite{DBLP:journals/jacm/MillerTTV97,DBLP:conf/focs/SmithW98} & circle & $3/4$ & $O(m^{3/4})$ & $O(n)$ & arbitrary disks\\
\cite{Fox20081070} & Jordan curve & $2/3$ & $O(\sqrt{m})$ & $\tilde{O}(n^2)$ & pseudodisks\\
Thm.~\ref{thm:joint_unit} & line & $2/3$ & $\tilde{O}(\sqrt{m})$ & $O(n)$& unit disks\\
Thm.~\ref{thm:axis-parallel} & axis-parallel line & $4/5$ & $O(\sqrt{m})$ & $O(n)$ & unit disks\\ \hline
\cite{DBLP:journals/dcg/AlonKP89,DBLP:journals/jocg/LofflerM14,DBLP:conf/esa/HoffmannKM14} & line& $1/2$& $O(\sqrt{n\log n})$& $O(n)$& {\bf disjoint} unit disks\\
\cite{DBLP:conf/esa/HoffmannKM14} & line & $1-\alpha$& $O(\sqrt{n/(1-2\alpha)})$&$O(n)$& {\bf disjoint} unit disks  \\
\cite{DBLP:journals/jocg/LofflerM14} & axis-parallel line & $9/10$ & $O(\sqrt{n})$ & $O(n)$& {\bf disjoint} unit disks \\
\end{tabular}
\label{tab:results-comp}
\end{table}

We emphasize that our results focus on \emph{unit} disk graphs, while the other results hold for disk graphs of arbitrary radii, too. Indeed, if we want to separate disks of arbitrary radii with a line, we show that the separator's size may be as large as $\Omega(n)$. We also prove that for unit disks our algorithm may fail to find a line $2/3$-separator of size better than $O(\sqrt{m \log(n/\sqrt{m})})$ in the worst case; the exact statement can be found in Section \ref{sec:lowerbound}. In this sense, the size of our separators is asymptotically almost tight. In Section \ref{sec:experiment}, experimental results are presented.
We evaluate the performance of our algorithm, compare it with
the method by Fox and Pach \cite{Fox20081070} 
in terms of the size of the produced separators for random instances, and conclude that our algorithm outperforms theirs for the intersection graphs of unit disks.

Working with a line separator for intersecting disks has some difficulty.
If we chose to separate pairwise disjoint geometric objects by a Jordan curve, 
then we could employ a volume argument for the interior of the curve.
However, we cannot use a volume argument for line separators since
the line does not determine a bounded region.

\paragraph{Other Related Work.}
In a different context, line separators of \emph{pairwise disjoint} unit disks have also been studied. Since the disks are pairwise disjoint, the intersection graph is trivially empty and can be easily separated. Instead, the focus is now to find a closed curve that intersects few disks, such that the two connected components it defines contain roughly the same number of disks. 

Alon~\etal~\cite{DBLP:journals/dcg/AlonKP89} 
proved that for a given set $\cal{D}$ of $n$ pairwise disjoint unit disks,\footnote{The result extends to pairwise disjoint fat objects that are convex and of similar area (see Theorem 4.1 of~\cite{DBLP:journals/dcg/AlonKP89}). For the sake of conciseness we only talk about unit disks.} there exists a slope $a$ such that 
every line with slope $a$ intersects $O(\sqrt{n \log n})$ unit disks
of $\cal{D}$. In particular, the halving line of that slope will be a nice separator (each halfplane will have at most $n/2-O(\sqrt{n \log n})$ disks fully contained in). Their proof is probabilistic, which can be turned into an expected $O(n)$-time randomized algorithm~\cite{DBLP:journals/jocg/LofflerM14}.
A deterministic $O(n)$-time algorithm was afterwards given by Hoffmann~\etal~\cite{DBLP:conf/esa/HoffmannKM14},
who also showed how to find a line $\ell$ that intersects at most $O(\sqrt{n/(1-2\alpha)})$ unit disks and 
each halfplane 
 contains at most $(1-\alpha)n$ disks
(for any $0 < \alpha < 1/2$). L\"offler and Mulzer \cite{DBLP:journals/jocg/LofflerM14}
proved that
there exists 
an axis-parallel line $\ell$ such that $\ell$ intersects $O(\sqrt{n})$ 
disks, and each halfplane
 contains at most $9n/10$ 
 unit disks, and such line can be found in $O(n)$ time.
Our result is more general (since it allows intersections), and has a better balancing parameter (their~$9/10$ versus our~$4/5$). 
For comparison purposes, these three results are also shown in Table~\ref{tab:results-comp}.

A significant amount of research has focused on the search for balanced line separators of unit disk graphs in the plane, but unlike the ones mentioned before no guarantee is given on the shape of the separator. Yan~\etal~\cite{DBLP:journals/comgeo/YanXD12} studied a separator
of unit disk graphs
for designing a low-delay compact routing labeling scheme for ad-hoc networks
modeled by unit disk graphs.
Their separator is a $2/3$-separator, but has no size guarantee.
Fu and Wang \cite{DBLP:journals/siamcomp/FuW07} studied the case where a 
unit disk graph is a $\sqrt{n}\times\sqrt{n}$ grid generated from a
regular grid, and proved that there exists a line $2/3$-separator of size at most $1.129\sqrt{n}$.
They used the obtained separator to give the first subexponential-time
algorithm for the protein folding problem in the HP model.
The bound of $1.129\sqrt{n}$ was afterwards improved by Fu~\etal~\cite{DBLP:journals/ijcga/FuOX08} to $1.02074\sqrt{n}$.
We note that it is not known whether a minimum-size $2/3$-separator
of a unit disk graph can be computed in polynomial time, although
the problem is known to be NP-hard for 
graphs of maximum degree three \cite{DBLP:journals/ipl/BuiJ92},
$3$-regular graphs \cite{DBLP:journals/computing/MullerW91}, 
and planar graphs \cite{DBLP:journals/jgaa/Fukuyama06}. Finally, Alber and Fiala~\cite{DBLP:journals/jal/AlberF04} studied the existence of separators for disk intersection graphs, but ask for additional constraints to the set of disks (such as requiring the disks to be at least $\lambda$ units apart, or bounding the ratio between the radii of the smallest and the largest disks).

\paragraph{Preliminaries.}
In this paper, all disks are assumed to be closed (i.e., the boundaries are part of the disks), and
a \emph{unit disk} has radius one (thus diameter two).
For a set $S$ of $n$ points in $\real^2$, there always exists a point $p \in \real^2$
such that every halfplane containing $p$ contains at least $n/3$ points
from $S$.
Such a point $p$ is called a \emph{centerpoint} of $S$, and can be found
in $O(n)$ time \cite{DBLP:journals/dcg/JadhavM94}.
Let $\ell$ be a line through a centerpoint of $S$.
Then, each of the two closed halfplanes bounded by~$\ell$ contains at least $n/3$ 
points of $S$,
which in turn means that 
each of the two open halfplanes bounded by $\ell$
contains at most $2n/3$ points of $S$ each.
Here, a halfplane $H$ (closed or open) \emph{contains} a point $p$ if
$p \in H$.
We also say a halfplane~$H$ \emph{contains} a disk $D$ if $D \subseteq H$.

\section{Upper Bounds}
\subsection{Pairwise Disjoint Unit Disks}

Let $\cal{D}$ be a set of $n$ unit disks in the plane.
We first consider the case where the disks in $\cal{D}$ are pairwise disjoint.
The results from this case will also be used for the more general case where the disks in $\cal{D}$
are not necessarily disjoint.
We note that the next lemma has a flavor similar to a theorem by Alon~\etal~\cite{DBLP:journals/dcg/AlonKP89}.

\begin{lemma}
\label{lem:disjoint_unit}
Let $\cal{D}$ be a set of $n$ pairwise disjoint unit disks in the plane and let~$p$ be any point in the plane.
Then the expected number of disks intersected by a random line through $p$ is $O(\sqrt{n})$.
\end{lemma}

\begin{proof}
Let $C_i$ be the circle of radius $2i$ centered at $p$, for $i=0,1,\ldots$.
Then each disk in ${\cal D}$ is intersected by at most two of these circles---if a disk is intersected by two circles the intersection takes place on the boundary of the disk.
Let ${\cal D}_i \subseteq {\cal D}$ be the set of disks that have non-empty intersection with $C_i$, for $i=0,1,\ldots$.
See \figurename~\ref{fig:disjoint_unit}.
Note that $\sum_{i} |{\cal D}_i| \leq 2n$.

\begin{figure}[t]
  \centering
  \begin{minipage}{0.31\textwidth}
    \centering
    \resizebox{\textwidth}{!}{\includegraphics{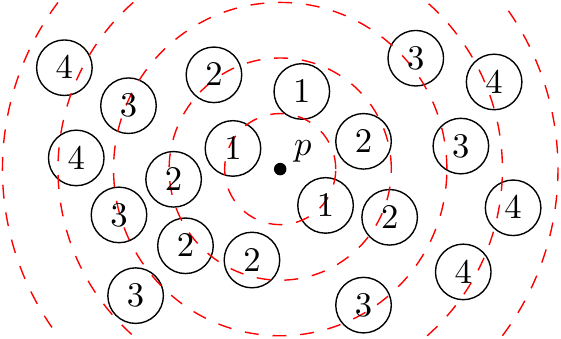}}
    \caption{Proof of Lemma \ref{lem:disjoint_unit}.
      The number $i$ in each disk means that it intersects the circle of radius $2i$ centered at $p$.
    }
    \label{fig:disjoint_unit}
  \end{minipage}
  \hfill
  \begin{minipage}{0.64\textwidth}
    \centering
    \resizebox{\textwidth}{!}{\includegraphics{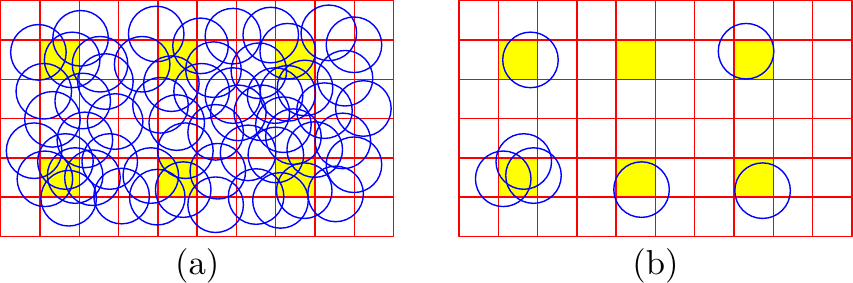}}
    \caption{Proof of Lemma \ref{lem:joint_unit}.
      (a) A grid of $\sqrt{2}\times \sqrt{2}$ squares is laid over the family of unit disks.
      (b) Disks associated to the same cell intersect, but disks associated to different cells of the same color do not intersect.
      If $j$ is the index for color yellow, then $n_j = 7$, $l_{j0}=4$, and $l_{j1}=1$ in this example.}
    \label{fig:joint_unit}
  \end{minipage}
\end{figure}

Given a random line $\ell$ through $p$, the number $k_i$ of disks of ${\cal D}_i$ that are intersected by $\ell$ is at most four, due to disjointness, and
its expectation is $O(|{\cal D}_i|/i)$.
Therefore, by the linearity of expectation,
the expected number of disks of $\cal{D}$ intersected by $\ell$ is bounded by
\begin{align*}
\Ex\left[\sum_{i \ge 0} k_i\right]
&=   \Ex\left[\sum_{i\colon i\leq \sqrt{n}} k_i\right] + \Ex\left[\sum_{i\colon i > \sqrt{n}} k_i\right]\\
&\le 4\sqrt{n} + \sum_{i\colon i > \sqrt{n}} O(|{\cal D}_i|/\sqrt{n})
 =   O(\sqrt{n}) + O(n/\sqrt{n}) = O(\sqrt{n}). 
\tag*{\qed}
\end{align*}
\end{proof}

\begin{corollary}
\label{cor:disjoint_unit}
Let $\cal{D}$ be a set of $n$ pairwise disjoint unit disks in the plane. 
Then, there exists a line $\ell$ that intersects $O(\sqrt{n})$ disks of $\cal{D}$ such that each of the two open halfplanes bounded by $\ell$ contains 
at most $2n/3$ disks of $\cal{D}$.
Moreover, such a line can be found in $O(n)$ time with probability
at least $3/4$.
\end{corollary}
\begin{proof}
Let $p$ be a centerpoint of the set of centers of disks in ${\cal D}$.
By Lemma~\ref{lem:disjoint_unit}, some line through $p$ must intersect
at most $O(\sqrt{n})$ disks from ${\cal D}$.
Since $p$ is a centerpoint, each of the two open halfplanes bounded by
$\ell$ contains at most $2n/3$ centers of disks in ${\cal D}$,
which means that the halfplane contains at most $2n/3$ disks from ${\cal D}$.

The argument above suggests the following algorithm:
first we compute a centerpoint~$p$ of the centers of a given set of
disks, and then choose a line through~$p$ uniformly at random.
By Lemma \ref{lem:disjoint_unit} and Markov's inequality, 
the probability that the random line intersects more than $c\sqrt{n}$ disks
is at most $1/4$ for some constant~$c$.
Thus, a desired line can be found with probability at least $3/4$.
The running time is linear in the number of disks since
a centerpoint can be found in linear time \cite{DBLP:journals/dcg/JadhavM94}.
\qed
\end{proof}

In the statement above, the fact that the probability of success is $3/4$ is not important (any positive probability would have also sufficed).

\subsection{General Case}
\label{subsec:generalcase}

We now consider the general case where the disks are not necessarily disjoint.

\begin{lemma}\label{lem:joint_unit}
Let $\cal{D}$ be a set of $n$ unit disks in the plane with $m$ intersecting pairs, and let $p$ be any point in the plane.
Then the expected number of disks intersected by a random line through $p$ is $O(\sqrt{(m+n)\log n})$.
\end{lemma}
\begin{proof}
Consider a grid of $\sqrt{2} \times \sqrt{2}$ squares. Each grid cell is treated as right-open and top-open so that it is of the form of $[x,x+\sqrt{2})\times [y,y+\sqrt{2})$.
Associate each disk in $\cal{D}$ with the grid cell containing its center, see \figurename~\ref{fig:joint_unit}(a).

Observe that one can color the grid cells with nine colors for every $3\times 3$ block of grid cells, so that no two disks that are associated with different grid cells of the same color intersect, see \figurename~\ref{fig:joint_unit}(b). 
Consider one of the colors $j$, with $1 \le j \le 9$, and let ${\mathbb C}_j$ be the collection of subsets of $\cal{D}$ associated with the grid cells of this color:
\[
{\mathbb C}_j = \{{\cal C} \subseteq {\cal D} \mid \text{ the center of disks in } {\cal C} \text{ lie in the same grid cell of color } j\}.
\]
Then, ${\mathbb C}_j$ has the following two properties:
(i) each subset ${\cal C} \in {\mathbb C}_j$ in the same grid cell is a clique, i.e., any two disks in the subset intersect each other;
(ii) any two disks from two different subsets in ${\mathbb C}_j$ are pairwise disjoint.
Let $n_j = \sum_{{\cal C} \in {\mathbb C}_j} |{\cal C}|$ denote the number of
disks in ${\cal D}$ associated to grid cells of color $j$.

We divide the cliques in ${\mathbb C}_j$ into $O(\log n_j)$ buckets ${\mathbb B}_{j0},{\mathbb B}_{j1},\ldots$, where ${\mathbb B}_{ji}$ consists of all cliques of ${\mathbb C}_j$ whose size is in the range $[2^i,2^{i+1})$. Set $l_{ji} = |{\mathbb B}_{ji}|$, for $i=0,1,\ldots$.
Then, the sum $x_{ji}$ of the sizes of the cliques in ${\mathbb B}_{ji}$ is in the range $l_{ji}2^i \le x_{ji} < l_{ji}2^{i+1} $.
We also know that $\sum_i x_{ji} = n_j$.
Let 
$m_{ji} = \sum_{{\cal C} \in \mathbb{B}_{ji}} |{\cal C}|(|{\cal C}|-1)/2$
denote the number of pairs of intersecting disks within ${\mathbb B}_{ji}$.
Then, $l_{ji} (2^{2i-1}-2^{i-1}) \le m_{ji} < l_{ji} (2^{2i+1}-2^{i})$ and
$\sum_i m_{ji} \le \sum_j \sum_i m_{ji} \leq m.$

We first compute the expected number of disks in cliques of ${\mathbb B}_{ji}$ intersected by a random line through $p$. Since the union of the disks in each clique is contained in a disk of radius 2 such that these disks of radius 2 are disjoint, a random line through $p$ intersects only $O(\sqrt{l_{ji}})$ cliques of ${\mathbb B}_{ji}$ by Lemma~\ref{lem:disjoint_unit}, and therefore only $O(\sqrt{l_{ji}}2^{i+1})$ disks of~${\cal D}$. By the definition of ${\mathbb B}_{ji}$, we have $l_{ji}2^{2i-1} \leq m_{ji} + l_{ji} 2^{i-1} \leq m_{ji} + x_{ji}/2$. Thus, a random line through $p$ intersects $O(\sqrt{m_{ji} + x_{ji}})$ disks in expectation.

Then, by the linearity of expectation,
we sum the numbers for all $j=1, \ldots, 9$ and $i = 0, \ldots, \log n_j$:
\begin{eqnarray*}
\sum_j \sum_i O(\sqrt{m_{ji} + x_{ji}})
 &\leq& O\left(\sqrt{\sum_j \sum_i (m_{ji} + x_{ji})} \sqrt{\sum_j \sum_i 1}\right)\\
 &\leq& O(\sqrt{(m+n)\log n}),
\end{eqnarray*}
where the first inequality follows from the Cauchy--Schwarz inequality. 
\qed
\end{proof}

In the same way as Corollary~\ref{cor:disjoint_unit} follows from Lemma~\ref{lem:disjoint_unit}, the following theorem follows from Lemma~\ref{lem:joint_unit}.

\begin{theorem}
\label{thm:joint_unit}
Let $\cal{D}$ be a set of $n$ unit disks in the plane with $m$ intersecting pairs.
Then, there exists a line $\ell$ that intersects $O(\sqrt{(m+n)\log n})$ disks
of ${\cal D}$ such that each of the two open halfplanes bounded by $\ell$ contains
at most $2n/3$ disks of ${\cal D}$.
Moreover, such a line can be found in $O(n)$ time with probability at least $3/4$.
\end{theorem}

As before, the exact value of $3/4$ for the success probability is not important.

We here note that Lemma \ref{lem:joint_unit} 
also holds 
even when the radii of all disks range between $1$ and $2$.
Consequently,  Theorem \ref{thm:joint_unit} also holds in such a case.
To adopt the proof of Lemma \ref{lem:joint_unit},
we change the coloring of the grid cells.
Namely, we color with 36 colors for every $6\times 6$ block of grid cells,
so that no two disks that are associated with different grid cells of the same
color intersect.
Then, the same argument works, and we obtain the lemma.

We may also generalize the argument when the radii of the disks are arbitrary.
Let $\cal{D}$ be a set of disks of arbitrary radii.
Then, we classify the disks in $\cal{D}$ by their radii.
By scaling, we assume that the smallest radius of the disks is one.
Let $\Delta$ be the radius of a largest disk (or, before scaling,
the ratio between the smallest and the largest radius of the disks).
Then, for each $i  = 1,2,\dots,\lceil \log_2 \Delta \rceil$,
the set ${\cal D}_i \subseteq \cal{D}$ consists of 
the disks whose radii lie between $2^{i-1}$ and $2^i$.
By the discussion above, the expected number of disks intersected by
a random line through an arbitrary point $p$ is $O(\sqrt{(m_i+n_i)\log n_i})$,
where $n_i=|{\cal D}_i|$ and $m_i$ is the number of intersecting pairs
of disks in ${\cal D}_i$.
Therefore, such a random line intersects $O(\sum_{i=1}^{\lceil \log_2\Delta\rceil} \sqrt{(m_i+n_i) \log n_i})$ disks in $\cal{D}$ in expectation.
Note that
\[
 \sum_{i=1}^{\lceil \log_2\Delta\rceil} \sqrt{(m_i+n_i) \log n_i})
 = O(\sqrt{(m+n)\log n}\log \Delta),
\]
where $n=|\cal{D}|$ and $m$ is the number of intersecting pairs of
disks in $\cal{D}$.
By choosing a  centerpoint of disk centers as $p$,
we obtain the following corollary.
\begin{corollary}
\label{cor:arb_radii}
Let $\cal{D}$ be a set of $n$ disks of arbitrary radii in the plane with $m$ intersecting pairs.
Let $\Delta$ be the ratio of the smallest and the largest radii of the disks in $\cal{D}$.
Then, there exists a line $\ell$ that intersects $O(\sqrt{(m+n)\log n}\log \Delta)$ disks
of ${\cal D}$ such that each of the two open halfplanes bounded by $\ell$ contains
at most $2n/3$ disks of ${\cal D}$.
Moreover, such a line can be found in $O(n)$ time with probability at least $3/4$.
\end{corollary}

\subsection{Axis-Parallel Separators}
In this section we show an alternative, more restricted separator. Specifically, we show that a line separator that intersects fewer disks ($O(\sqrt{m+n})$) exists, even if we restrict ourselves to axis-parallel lines. However, this comes at the cost that the balancing parameter is worsened:
on each side of the line we can certify only that there are at most $4n/5$ disks. 
The theorem below and its proof have a flavor similar to those by L\"offler and Mulzer \cite{DBLP:journals/jocg/LofflerM14}.

\begin{theorem}
\label{thm:axis-parallel}
Let $\cal{D}$ be a set of $n$ unit disks in the plane with $m$ intersecting pairs.  
Then, there exists an axis-parallel line $\ell$ that intersects $O(\sqrt{m+n})$ disks of $\cal{D}$ such that each of the two open halfplanes bounded by $\ell$ contains at most $4n/5$ disks of $\cal{D}$.
Moreover, such a line can be found in $O(n)$ time.
\end{theorem}

Before giving the proof, we first need to state a simple fact (whose proof follows from elementary analysis). 
\begin{lemma}
\label{lem:uniform}
Let $n,k$ be positive integers, and 
$f(x_1, \dots, x_k) = \sum_{i=1}^k x_i(x_i-1)/2$.
Then, $\min\{f(x_1, \dots, x_k) \mid (x_1,\dots,x_k) \in \mathbb{Z}^k, \sum_{i=1}^k x_i = n\} \geq n^2/(2k) - n/2$.
\end{lemma}

\begin{proof}[of Theorem \ref{thm:axis-parallel}]
Let $P$ be the set of disk centers and assume that there are no two points (centers) in $P$ with the same $x$-coordinate or the same $y$-coordinate. Let $\ell_d$ (resp., $\ell_u$) be a horizontal line such that there are exactly $n/5$ points of $P$ below it (resp., above it).
Let $H$ be the distance between~$\ell_d$ and~$\ell_u$.
Similarly, let $\ell_l$ (resp., $\ell_r$) be a vertical line such that there are exactly $n/5$ points of $P$ to its left (resp., to its right).
Let $V$ be the distance between $\ell_l$ and $\ell_r$. 
{From} the above definitions, it follows that any horizontal line $\ell$ between $\ell_d$ and $\ell_u$, and any vertical line $\ell$ between $\ell_l$ and $\ell_r$ have at most $4n/5$ disks of $\cal{D}$ on each of the two open halfplanes bounded by $\ell$. We will show that one of these lines intersects $O(\sqrt{m+n})$ disks of $\cal{D}$.

Consider the rectangle $\cal{R}$ defined by $\ell_d$, $\ell_u$, $\ell_l$ and $\ell_r$.
Now let $h_i$ be the horizontal line above $\ell_d$ whose distance from $\ell_d$ is exactly $i$, for $i=1,\ldots,\lceil H-1\rceil$, and similarly,
let $v_i$ be the vertical line to the right of $\ell_l$ whose distance from $\ell_l$ is exactly $i$, for $i=1,\ldots,\lceil V-1 \rceil$.

We have a case analysis.
Assume first that $H < 2$ and $V<2$.
Then, there exist a horizontal line between $\ell_d$ and $\ell_u$, and a vertical line between $\ell_l$ and $\ell_r$ that intersect at least $n/5$ disks whose centers are in $\cal{R}$. What we need to show is $m=\Omega(n^2)$ as any line intersects at most $O(n)$ disks. We partition $\cal{R}$ into at most $4$ rectangles by $h_i$ and $v_i$ so that each side of these rectangles is at most $\sqrt{2}$. By Lemma~\ref{lem:uniform} with $k \leq 4$, we have $m \geq n^2/200 - n/10 = \Omega(n^2)$.

Next we assume that $H \geq 2$ and $V \geq 2$.
Then, no disk with center below $\ell_d$ or above $\ell_u$ (resp., left to $\ell_l$ or right to $\ell_r$) intersects $h_i$ for $i=1, \ldots, \lfloor H-1 \rfloor$ (resp., $v_i$ for $i=1, \ldots, \lfloor V-1 \rfloor$).  The number of centers in $\cal{R}$ is at least $n/5$ and any disk whose center is in $\cal{R}$ is intersected by 
at most three of the horizontal lines $h_i$ and by 
at most three of the vertical lines $v_i$.

Now, if one (or more) of the lines $h_i$ or one (or more) of the lines $v_i$ intersects at most $\lambda \sqrt{m+n/10}$ disks, for some constant $\lambda$ that we specify below, then we are done.
Assume therefore that each of the lines $h_i$ and each of the lines $v_i$ intersects more than $\lambda \sqrt{m+n/10}$ disks.
Since no disk with center below $\ell_d$ or above $\ell_u$ (resp., left to $\ell_l$ or right to $\ell_r$) intersects $h_i$ for $i=1, \ldots, \lfloor H-1 \rfloor$ (resp., $v_j$ for $j=1, \ldots, \lfloor V-1 \rfloor$), 
we have 
\[
\lambda \sqrt{m+n/10} \lfloor H-1 \rfloor  \leq 3\left(\frac{3n}{5}\right), 
\quad
\text{and}
\quad
\lambda \sqrt{m+n/10} \lfloor V-1 \rfloor  \leq 3\left(\frac{3n}{5}\right);
\]
otherwise we get that the number of disks is greater than $n$. 

We partition $\cal{R}$ into $N$ rectangles 
using $h_i$ and $v_j$ for $1 \leq i \leq \lceil H-1 \rceil$ and $1 \leq j \leq \lceil V-1 \rceil$
so that each side of these rectangles is at most $\sqrt{2}$. 
Then, the disk centers in $\cal{R}$ are distributed into $N$ rectangles, and the disks with centers in the same rectangle form a clique.
Let $x_i$ be the number of centers in one of the $N$ rectangles, for $i=1,\ldots, N$, where $\sum_{i=1}^N x_i \geq n/5$.
By Lemma~\ref{lem:uniform} with $k=N$, the number of intersecting pairs of disks whose centers lie in $\cal{R}$ is at least $n^2/(50N) - n/10$.  Then, we have 
\[
m + \frac{n}{10} \geq \frac{n^2}{50N}.
\]
On the other hand, we have 
\[
N = \lceil H \rceil \lceil V \rceil \leq (\lfloor H -1 \rfloor + 2) (\lfloor V -1 \rfloor + 2) \leq 9 (\lfloor H -1 \rfloor) (\lfloor V -1 \rfloor)
\leq 81 \left(\frac{3n}{5}\right)^2 \cdot \frac{1}{\lambda^2 (m+n/10)}.
\]
Therefore, we conclude that
\[ m + \frac{n}{10}
\geq \frac{n^2}{50} \cdot \frac{25}{729n^2} \cdot \lambda^2 \left(m+\frac{n}{10}\right)
= \frac{1}{1458} \lambda^2 \left(m+\frac{n}{10}\right)
> m+\frac{n}{10},
\]
assuming $\lambda > \sqrt{1458}$.
This is a contradiction.

Finally we assume that $H \geq 2$ and $V < 2$ (the case where $V \geq 2$ and $H < 2$ is symmetric).
We argue as in the previous case.
If one (or more) of the lines $h_i$ intersects at most $\lambda \sqrt{m+n/10}$ disks, for some constant $\lambda$ that we specify below, then we are done.
Assume therefore that each of the lines $h_i$ intersects more than $\lambda \sqrt{m+n/10}$ disks.
Then, we have 
\[
\lambda \sqrt{m+n/10} \lfloor H-1 \rfloor  \leq 3\left(\frac{3n}{5}\right),
\]
and if we partition $\cal{R}$ into $N$ rectangles 
using $h_i$ and $v_j$ for $1 \leq i \leq \lceil H-1 \rceil$ and $j \leq \lceil V-1 \rceil$
so that each side of these rectangles is at most $\sqrt{2}$, then we have
\[
N 
\leq 2 \lceil H \rceil  \leq 2 (\lfloor H -1 \rfloor + 2) \leq 6 (\lfloor H -1 \rfloor) \leq \frac{54n}{5\lambda\sqrt{m+n/10}}.
\]
Again, as in the previous case, we have
\[
m + \frac{n}{10} \geq \frac{n^2}{50N} \geq \frac{n^2}{50} \cdot \frac{5 \lambda \sqrt{m+\dfrac{n}{10}}}{54n} >  \frac{\lambda\sqrt{2}}{540} \cdot \frac{n \sqrt{m+\dfrac{n}{10}}}{\sqrt{2}} > \frac{\sqrt{2}\lambda}{540} \left(m +\frac{n}{10}\right) > m+\frac{n}{10},
\]
assuming $\lambda > 540/\sqrt{2}$.
The second to last inequality follows from the upper bound on $m$ in the complete graph, i.e.\ $m \leq n(n-1)/2 \Rightarrow \sqrt{m+n/10} < n/\sqrt{2}$.
Thus, we have reached a contradiction.
\qed
\end{proof}

\section{Almost Tightness of Our Approach}
\label{sec:lowerbound}

We now show that the approach used in Theorem~\ref{thm:joint_unit} 
with a centerpoint
cannot be drastically improved. 
Specifically, we present a family~${\cal D}$ of $n$ unit disks
with a centerpoint~$p$ of the centers of those unit disks
such that any line that passes through $p$
will intersect many disks. 
Although this example can be constructed for any number of disks $n>0$ and 
any desired number of intersecting pairs $m>n$, the details are a bit tedious.
Instead, given the desired values $n$ and $m$, we find $n'\approx n$ and $m'\approx m$ that satisfy the properties. This greatly simplifies the proof and, asymptotically speaking, the bounds are unaffected.

\begin{theorem}\label{thm:lowerjoint_unit}
For any $n,m\in \mathbb{N}$ such that $9n \leq m \leq \lfloor n^2/6 \rfloor$,
there exist $n', m' \in \mathbb{N}$, where $n \leq n' \leq 2n$ and $\lceil m/9 \rceil \leq m' \leq 6m$,
and a set $\cal{D}$ of $n'$ unit disks in the plane with $m'$ intersecting pairs
that have the following property.
There exists a centerpoint $p$ of the centers of unit disks in $\cal{D}$
such that any line $\ell$ that passes through $p$
intersects $\Omega( \sqrt{m\log (n/\sqrt{m})})$ disks of $\cal{D}$.
\end{theorem}
\begin{proof}
Given $n$ and $m$ we choose $k\in \mathbb{N}$ as the smallest natural number 
$k'$ such that $k' \geq \sqrt{\frac{{6m}}{{1+\ln (n/k')}}}$.
First observe that such $k$ exists and satisfies $k\leq n$. 
Indeed, let $f(x)=\sqrt{\frac{6m}{1+\ln{(n/x)}}}$. 
Since $f(x)$ is monotonically increasing for $x \geq 1$, 
it suffices to show that $n-f(n) \geq 0$. This holds as
$$
n-f(n)  \geq 0 
\Leftrightarrow 
n-\sqrt{6m}  \geq 0
\Leftrightarrow
 n^2/6 \geq m,$$
and the last inequality is true by our assumed bounds on $m$.

Let $\ell=\lceil n/k \rceil$, and consider a sufficiently
small positive real number \mbox{$\varepsilon \leq \frac{1}{2\pi}$}.
Consider now the $\ell$ concentric circles $C_i$ centered at the origin with radius \mbox{$2i(1+\varepsilon)$} for $i=1, \dots, \ell$.
On each such circle, we place $k$ unit disks uniformly (i.e., the arc spacing between the centers of two consecutive disks is $4\pi i(1+\varepsilon)/k$).
Let $\cal{D}$ be the collection of these disks, see \figurename~\ref{fig:lowerBound}, left. 

Let $n' = k\ell$.
By construction, $\cal{D}$ has $n'$ disks,
and $n \leq n' \leq n+k \leq 2n$ as claimed.
Partition the disks of $\cal{D}$ into $\ell$ layers ${\cal D}_1, \dots, {\cal D}_{\ell}$ depending on which concentric circle their center lies on.
Since we placed unit disks on circles that are $2(1+\varepsilon)$ units apart, only disks that belong to the same layer may have nonempty intersection.
We now show that the number $m'$ of intersecting pairs in $\cal{D}$
satisfies $\lceil m/9 \rceil \leq m' \leq 6m$.

Let $x_i$ denote the number of intersecting pairs of disks on the $i$-th layer, 
and let $\gamma_i$ be the 
 arc length of $C_i$ such that two unit disks centered at two endpoints of the arc touch each other (see \figurename~\ref{fig:lowerBound}, middle).
Note that two disks on the same layer overlap if and only if the arc distance between the centers is $\gamma_i$ or less, and disks in different layers cannot intersect. 

Since the disks are evenly spaced on $C_i$, each disk of ${\cal D}_i$ intersects $2\lfloor \frac{\gamma_i}{4\pi i(1+\varepsilon)/k} \rfloor $ other disks. Taking into account that there are $k$ disks in ${\cal D}_i$ (and each crossing is counted twice), we obtain that $x_i= \lfloor\frac{\gamma_ik}{4\pi i(1+\varepsilon)}\rfloor k$ and thus the total number of intersecting pairs is $m'=\sum_{i=1}^{\ell} x_i = \sum_{i=1}^{\ell} \lfloor\frac{\gamma_ik}{4\pi i(1+\varepsilon)}\rfloor k$.

By using some trigonometric properties and picking a sufficiently small $\varepsilon$, we can show that, regardless of the value of $i$, we have 
$$2 < \gamma_i = 4i(1+\varepsilon)\sin^{-1}(1/2i(1+\varepsilon)) < 2\pi/3.$$

We now bound $m'$ from above.
Since $\gamma_i< 2\pi/3$, we have 
$$m' < \sum_{i=1}^{\ell} \frac{ k^2}{6 i(1+\varepsilon)} < \frac{k^2}{6}\sum_{i=1}^{\ell} \frac{1}{i} < k^2(1+\ln \ell)/6 < k^2(1+\ln (2n/k))/6.$$ Recall that we choose $k$ to be the smallest natural number $k'$ such that $k' \geq \sqrt{\frac{{6m}}{{1+\ln (n/k')}}}$. 
We combine this bound with the assumption that $n<m$ to obtain the desired bound:
\begin{align*}
 m' &< k^2\frac{(1+\ln (2n/k))}{6} 
     < \left(1+\sqrt{\frac{{6m}}{{1+\ln (n/k)}}}\right)^2 \left(\sqrt{\frac{(1+2\ln (n/k))}{6}}\right)^2 \\
    &< \left(\sqrt{\frac{1+2\ln (n/k)}{6}}+\sqrt{2m} \right)^2 <\left(\sqrt{\frac{6n}{6}}+\sqrt{2m} \right)^2 \\
    &<  (\sqrt{m}+\sqrt{2m})^2 =(3+2\sqrt{2}) m.
\end{align*}

For the lower bound we now use that $\gamma_i>2$ (and $\varepsilon \leq \frac{1}{2\pi}$). Recall that  
\begin{align*}
m' &= \sum_{i=1}^{\ell} \left\lfloor\frac{\gamma_ik}{4\pi i(1+\varepsilon)}\right\rfloor k >  \sum_{i=1}^{\ell} \left( \frac{k}{2\pi i(1+\varepsilon)} -1\right) k 
 >  \frac{k^2}{2\pi + 1} \left(\sum_{i=1}^{\ell} \frac{1}{i}\right)
 - \left(\sum_{i=1}^{\ell} k\right) > \frac{k^2 \ln(l+1)}{2\pi + 1} - n'.
\end{align*}
 Recall that, by choice of $k$, we have $k \geq \sqrt{\frac{{6m}}{{1+\ln (n/k)}}}$ and in particular $k \geq \sqrt{\frac{6m}{1+\ln \ell}}$. Combining this fact with $\ln(\ell+1) \geq (1+\ln \ell)/2$ and $2m/9 \geq n'$ gives us the lower bound as follows:
$$
m' \geq \left(\frac{6m}{1+\ln \ell }\right) \left(\frac{\ln(\ell+1)}{2\pi + 1}\right) - n' 
    \geq \frac{3m}{2\pi + 1} - \frac{2m}{9} > \frac{3m}{9} - \frac{2m}{9}=\frac{m}{9}.
$$
That is, for any values of $n$ and $m \in \{9n,\dots,\lfloor n^2/6 \rfloor\}$, we can construct a set of $n'$ unit disks with $m'$ intersecting pairs.

By symmetry, we can see that
the origin is a centerpoint of the centers of disks in $\cal{D}$.
Thus, it remains to show that any line that passes through the origin must intersect many disks of ${\cal D}$. In the following we show something stronger:
any ray emanating from a point $p$ inside $C_1$ will cross $\Omega( \sqrt{m\log (n/\sqrt{m})})$ disks.

We count the number of intersections on each layer independently. Each unit disk in ${\cal D}_i$ covers $\frac{\gamma_i}{2i(1+\varepsilon)}$ radians of $C_i$. Since there are $k$ evenly spread disks in ${\cal D}_i$ and $\varepsilon \leq \frac{1}{4\pi}$, each point of the circle $C_i$ is contained in at least $\lfloor \frac{\gamma_i k}{2i(1+\varepsilon)} \cdot \frac{1}{2\pi} \rfloor > \frac{\gamma_i k}{i(4\pi +1)} - 1$ disks.

By substituting $k \geq \sqrt{\frac{6m}{1+\ln \ell}}$, $\gamma_i > 2$, $\ln(\ell+1) \geq (1+\ln \ell)/2$, $\ell = n'/k$ and $2m/9 \geq n'$, we obtain that the number of disks in ${\cal D}$ intersected
by the ray is at least
\begin{align*}
\sum_{i=1}^{\ell} (\frac{\gamma_i k}{i(4\pi +1)} - 1) 
 & >    \frac{2}{4\pi + 1}\left(\sqrt{\frac{6m}{1+\ln \ell }}\right)\left( \sum_{i=1}^{\ell} \frac{1}{i}\right) - \ell\\
 & \geq \frac{2}{4\pi + 1} \left(\sqrt{\frac{6m}{1+\ln \ell}}\right)\ln(\ell +1) - \frac{n' \sqrt{1+\ln \ell}}{\sqrt{6m}}\\
 & \geq \frac{\sqrt{6m (1+\ln \ell )}}{4\pi + 1} - \frac{2m}{9} \sqrt{\frac{1+\ln \ell}{6m}}\\
 & \geq \left(\frac{\sqrt{6}}{4\pi + 1} - \frac{2}{9\sqrt{6}}\right) \sqrt{m(1+\ln \ell)} =  \Omega( \sqrt{m\log (n/\sqrt{m})}).
\tag*{\qed}
\end{align*}
\end{proof}

\begin{figure}[t]
  \centering
  \scalebox{.4}{\includegraphics[page=3]{lowerBound}}
  \qquad
  \scalebox{.5}{\includegraphics[page=2]{lowerBound}}
  \qquad
  \scalebox{.25}{\includegraphics{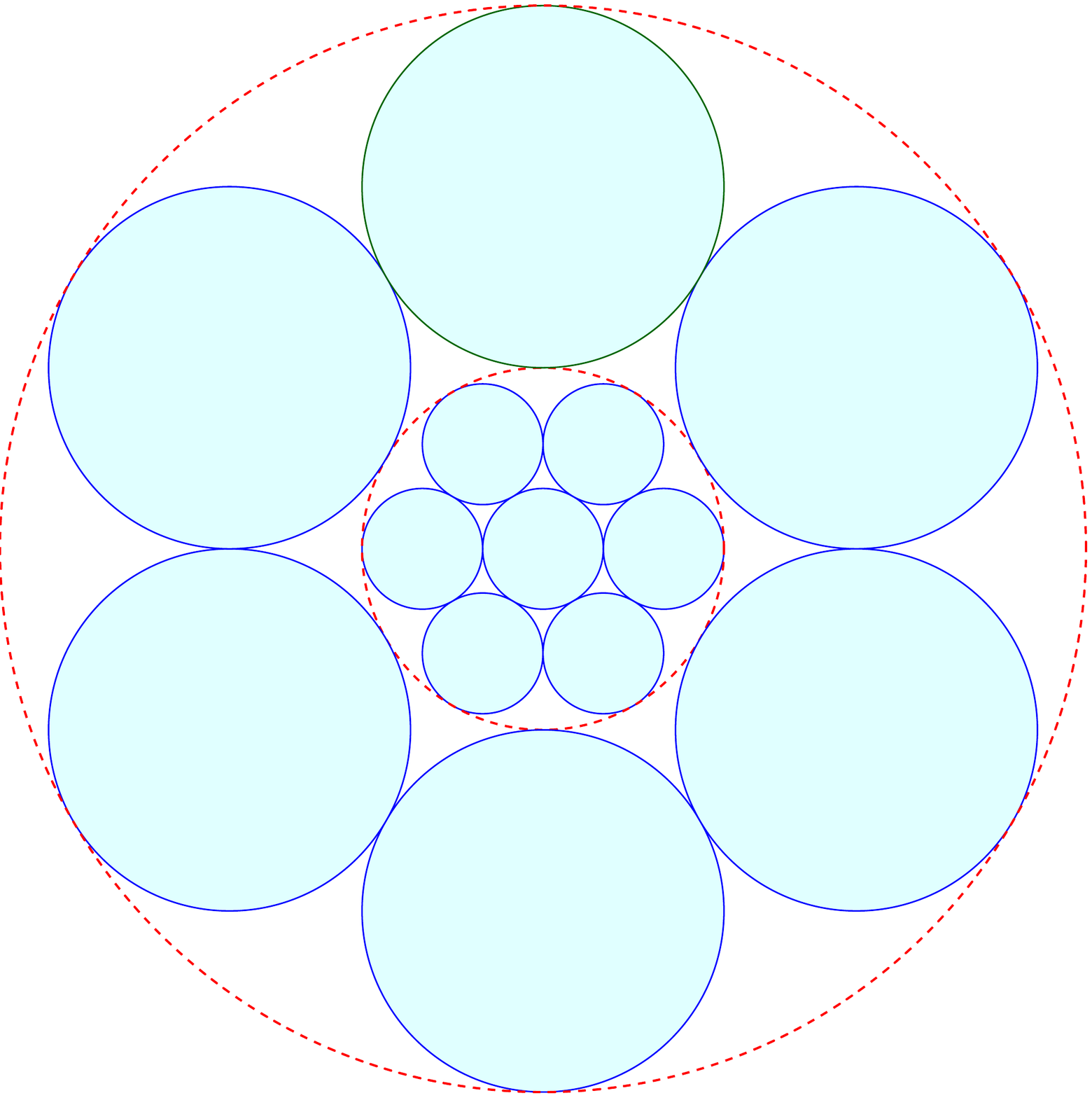}}
  \caption{
    Left: Almost tightness construction for $\ell=3$ (and $k=16$). 
    Middle: upper and lower bounds for $\gamma_i$. 
    Right: construction for disks of arbitrary radii.
  }
  \label{fig:lowerBound}
\end{figure}

We have shown in Theorems \ref{thm:joint_unit} and \ref{thm:axis-parallel} that a good line separator can be found for unit disks. However, we show next that this is not possible when considering disks of arbitrary sizes.

\begin{theorem}\label{thm:lowerarbitRad}
For any $n\in \mathbb{N}$ there exists a set $\cal{D}$ of $O(n)$ pairwise disjoint disks such that any line $\ell$ for which both halfplanes contain $\Omega(n)$ disks of $\cal{D}$ also 
intersects $\Omega(n)$ disks of $\cal{D}$.
\end{theorem}
\begin{proof}
Our construction is similar to the one of Theorem~\ref{thm:lowerjoint_unit}: let ${\cal C}_1$ be the set containing the unit disk centered at the origin. For $i>1$, let $D_i$ be the smallest enclosing disk containing all disks of ${\cal C}_{i-1}$. The set ${\cal C}_i$ (for $i>1$) is defined as the set containing the six disks of the same radius as $D_i$ surrounding $D_i$, such that each two consecutive disks touch and none overlap. Note that $D_i\not\in {\cal C}_{i}$ and ${\cal C}=\bigcup_{i=1}^n {\cal C}_i$. 

Observe that in this construction disks of one level touch only neighbors on the same level (and on the next and previous level). We fix this by rotating the construction so that disks in ${\cal C}_i$ are tangent to $D_i$ but do not intersect with the disks of ${\cal C}_{i-1}$ (see Figure~\ref{fig:lowerBound} right). Tangencies can afterwards be removed by shirking the disks by a factor of $1-\varepsilon$ for a sufficiently small value of $\varepsilon$. We also shrink the disk in ${\cal C}_1$ by this factor of $1-\varepsilon$. After these two changes the disks will be pairwise disjoint.

The argument is now similar to the one in Theorem~\ref{thm:lowerjoint_unit}, and the centers of the disks in $\cal C$ are placed in concentric circles around the origin. The main difference now is that each layer has a constant number of disks. Thus, any line $\ell$ that makes a balanced separator must cross $\Omega(n)$ concentric circles. For every two intersections we can find at least one disk that intersects with $\ell$, and thus the $\Omega(n)$ lower bound follows. \qed
\end{proof}

Note that in this construction the radii of the disks grow at an exponential rate. 
Therefore, our upper bound in Corollary \ref{cor:arb_radii} is nearly optimal
also in this sense.

\section{Experiments}
\label{sec:experiment}

In our experiments, we evaluate the quality of separator algorithms
by their separator size. Theorem~\ref{thm:joint_unit} suggests a simple algorithm: find a centerpoint (which can be done in linear time~\cite{DBLP:journals/dcg/JadhavM94}) and try random lines passing through that point until a good separator is found. Since implementing the centerpoint algorithm is not trivial, we use an alternative method that is asymptotically slower but much easier to implement: for a slope $a$ selected uniformly at random find a $2/3$-separator with slope $a$ that intersects the minimum number of disks (this step can be done in $O(n \log n)$ time by sorting the disks in orthogonal direction of $a$ and making a plane sweep). 
 Clearly, a separator found by the modified algorithm intersects at most as many disks as the line of the same slope passing through a centerpoint. Thus, as in Theorem~\ref{thm:joint_unit}, a random direction will be good with positive probability. 
We repeat this process many times to obtain a good line separator.

We compare our algorithm with the method by Fox and Pach \cite{Fox20081070}
which guarantees the separator size of $O(\sqrt{m})$. For the implementation of our algorithm we use the simpler variation described above. 

\subsection{The Method by Fox and Pach}
Fox and Pach \cite{Fox20081070} proved that the intersection graph
of a set of Jordan curves in the plane has a $2/3$-separator of size
$O(\sqrt{m})$ if every pair of curves intersects in a constant number 
of points.
Their proof is constructive, as outlined below.

First, we build the arrangement of curves, and obtain a plane graph whose vertex set are the vertices of the arrangement and consecutive vertices on a curve are joined
by an edge.\footnote{The method of Fox and Pach needs to add a constant number of additional vertices, but the main feature is that the overall complexity of the graph is $O(m+n)$.}
We triangulate the obtained plane graph to make it maximal
planar.
Then, we find a simple cycle $2/3$-separator $C$ 
(i.e., a $2/3$-separator that forms a cycle in the graph)
of size $O(\sqrt{m+n})$,
which always exists \cite{DBLP:journals/jcss/Miller86}.
We output all curves containing a vertex in $C$.

In our implementation, we construct the circle arrangement in a brute-force manner, and we use a simple cycle separator algorithm
by Holzer~\etal~\cite{DBLP:journals/jea/HolzerSWPZ09}, called the
fundamental cycle separator (FCS) algorithm.
Although the FCS algorithm has no theoretical guarantee for the size
of the obtained separator, the recent experimental study by 
Fox-Epstein~\etal~\cite{DBLP:journals/jea/Fox-EpsteinMP016} showed 
that it has a comparable
performance to the state-of-the-art cycle separator algorithm with
theoretical guarantee for most of the cases.

\subsection{Instance Generation and Experiment Setup}
We use two sets of instances for our experiments.
We call the first set a set of \emph{random instances}
and the second set a set of \emph{snake instances}.
Random instances are generated at random, 
as usually done for the experimental work on sensor networks.
We fix a square $S$ of side length $L$ and generate
$n$ unit disks in $S$ independently and uniformly at random.
If the graph is disconnected, we discard it and generate again.

The snake instance is an instance designed to be particularly challenging for our algorithm. Intuitively speaking, for any odd $n$ it places $n$ disks in a square of sizelength $\sqrt{2n}$ (a formal description follows below, see an example instance in  \figurename~\ref{fig:snake}). The disks are placed in a way that their intersection graph is a path and thus sparse. A particular trait of this instance is that it has a separator of constant size: a vertical line can do a balanced cut and only intersect one disk. However, the separator will be very hard to find when we pick a random line. The probability tends to zero as the size of the instance grows. 
The purpose of snake instances is to observe the behavior of the algorithms
in adverse conditions.

The specific construction of the snake instance is as follows: for ease of generation, we set the radius of the unit disks to $4/3$, and
we fix an odd integer $q$.
Let $n=(q^2-1)/2 + q$.
We place $n$ unit disk centers at
integer coordinates $(x,y) \in \{2i-1 \mid 1\leq i \leq (q+1)/2\}
\times \{j \mid 1\leq j \leq q\} \cup \{(2i, q^{i \bmod 2}) \mid 1\leq i\leq (q-1)/2 \}$.
Note that 
$q^{i \bmod 2}$ is equal to $1$ when $i$ is even, and is equal to $q$ when
$i$ is odd.
Hence, for each $q$, the snake instance is uniquely determined.
Observe that, by constructions, the number $m$ of edges is $n-1 = (q^2-1)/2 + q-1$.

\begin{figure}[t]
  \centering
  \scalebox{0.5}{\includegraphics{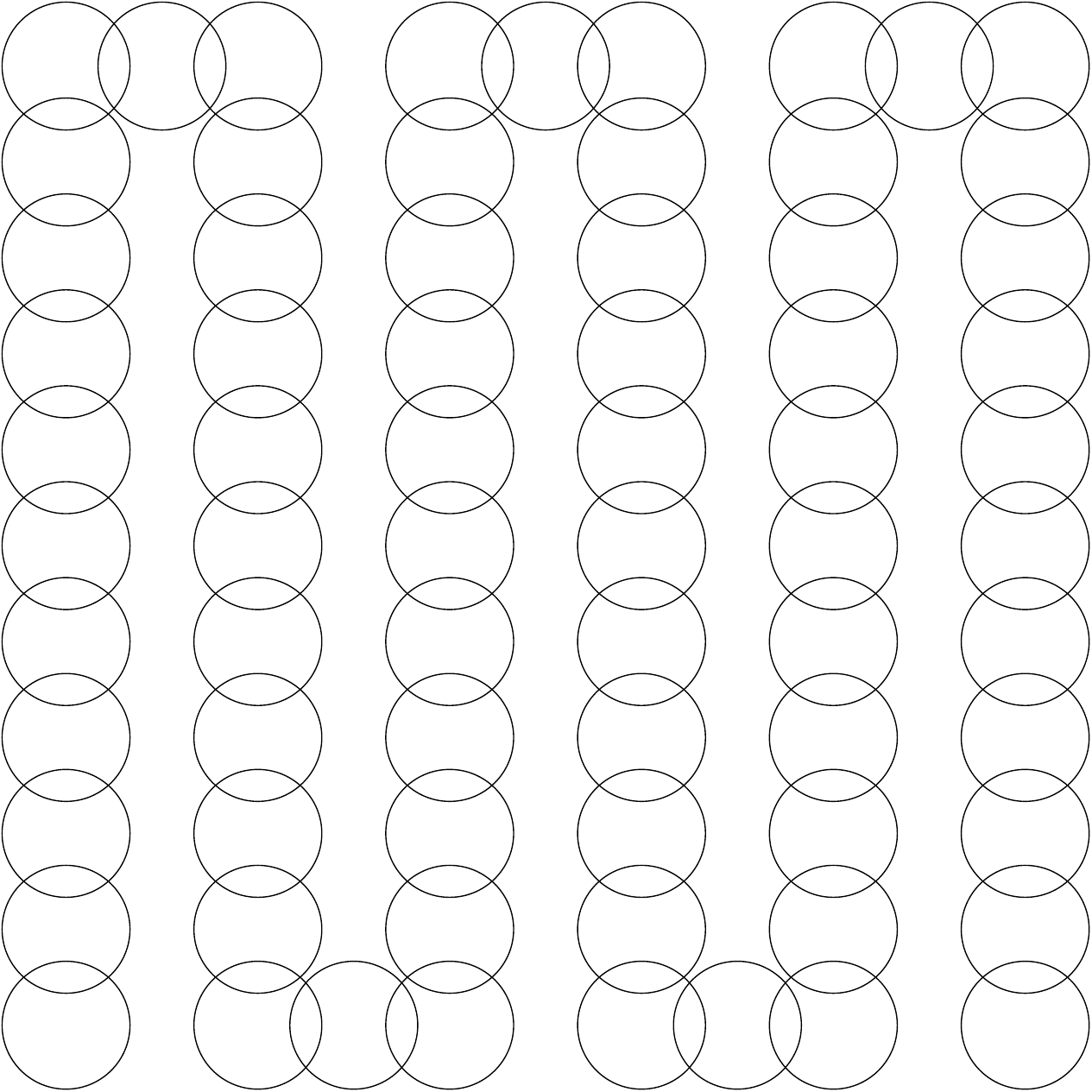}}
  \caption{
    The snake instance with $q=11$.
    In this case, $n=(q^2-1)/2+q = 71$. Therefore, the instance
    contains $71$ unit disks.
  }
  \label{fig:snake}
\end{figure}

All experiments have been performed on Intel $\circledR$ Core$^{\text{TM}}$
i7-5600U CPU @2.60GHz $\times$ 4, with 7.7GB memory and
976.0GB hard disk, running Ubuntu 14.04.3 LTS 64bit.

\subsection{Experiment 1: Quality of the Proposed Method}
In the first experiment we empirically examine the size of a separator
obtained by our proposed algorithm
with the modification proposed at the beginning of this section.

For random instances, we fix $L=100$, and vary the number of disks $n$ 
from $10,000$ to $30,000$ with an increment of $50$.
Since our algorithm is randomized, we run the algorithm $k$ times, where
$k\in \{1, 2, 10, 15, 20\}$, and compute the average separator size (because of Theorem~\ref{thm:joint_unit}, we expect the average to converge to $O(\sqrt{m\log n})$ as $k$ grows to infinity).

\begin{figure}[t]
  \centering
  \resizebox{.4\textwidth}{!}{\includegraphics{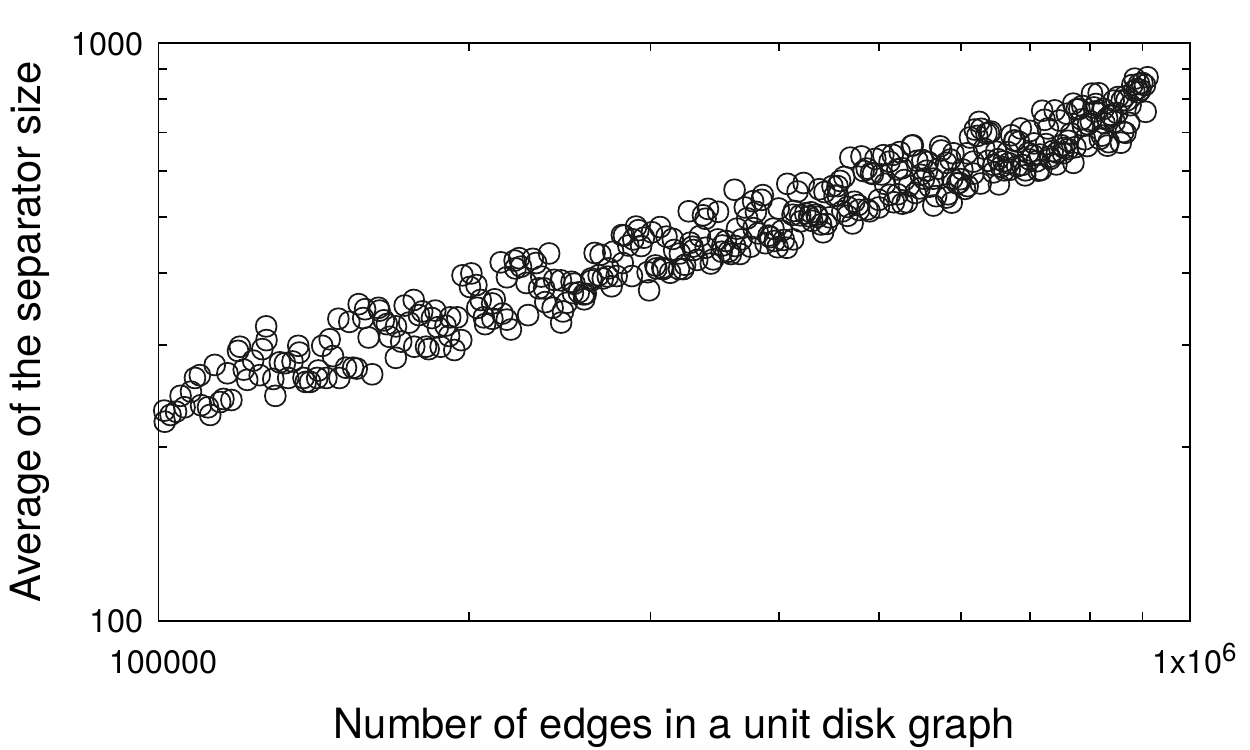}}
  \resizebox{.4\textwidth}{!}{\includegraphics{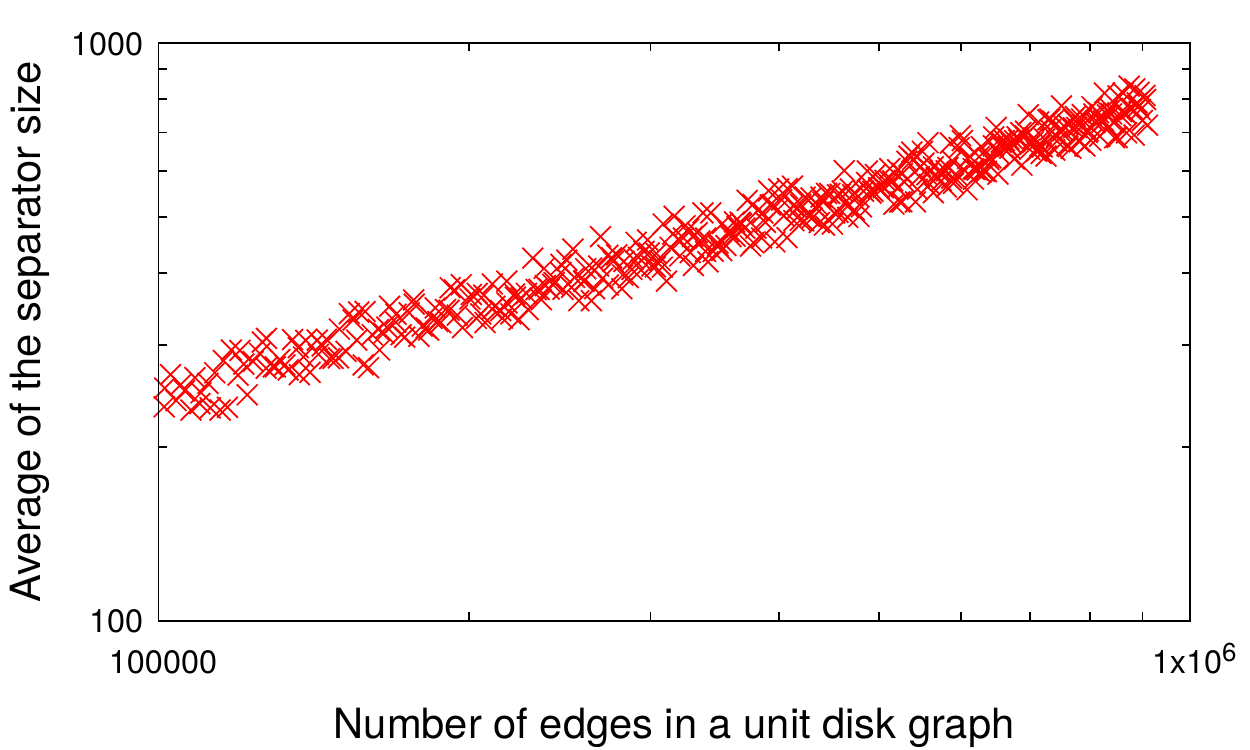}}
  \resizebox{.4\textwidth}{!}{\includegraphics{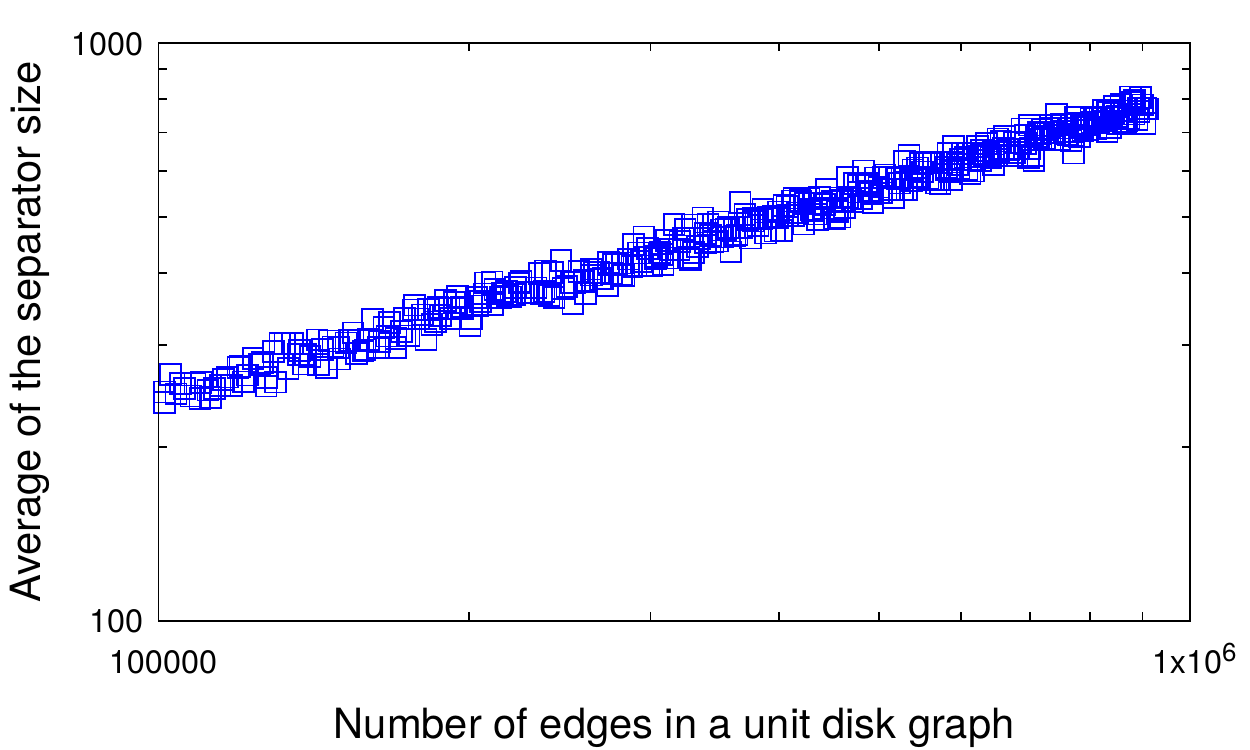}}
  \resizebox{.4\textwidth}{!}{\includegraphics{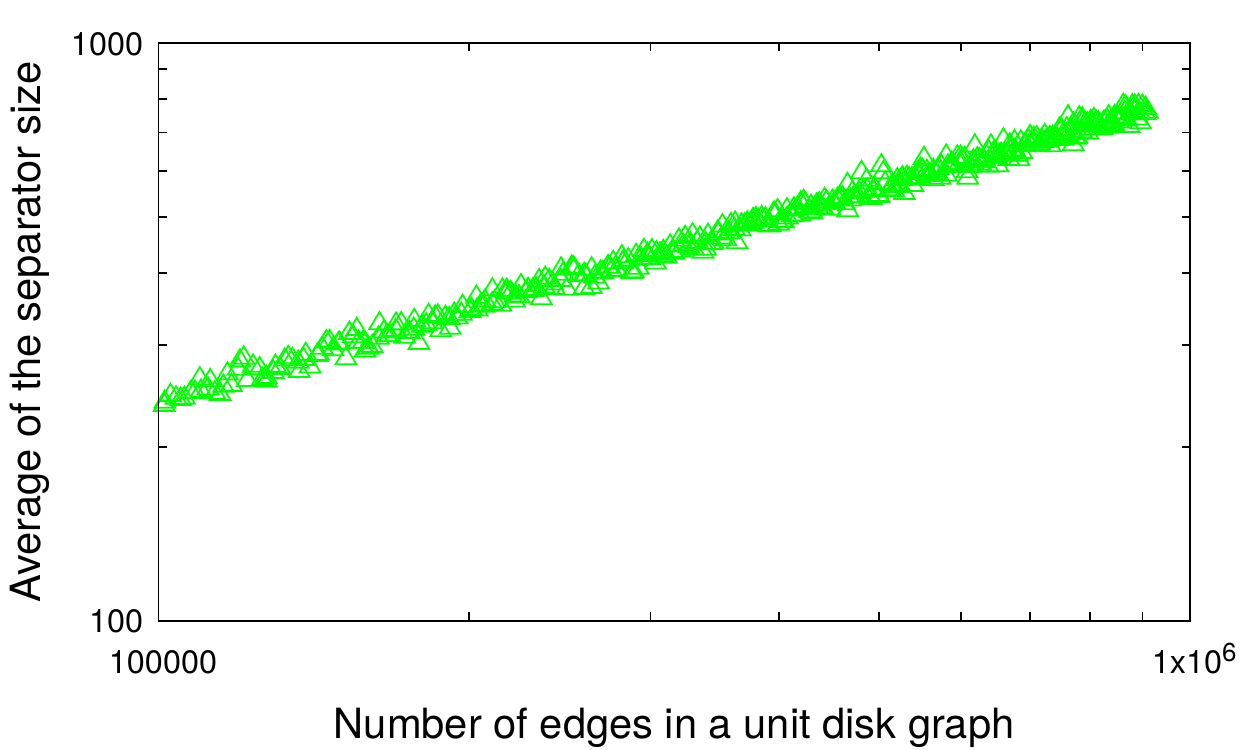}}
  \resizebox{.4\textwidth}{!}{\includegraphics{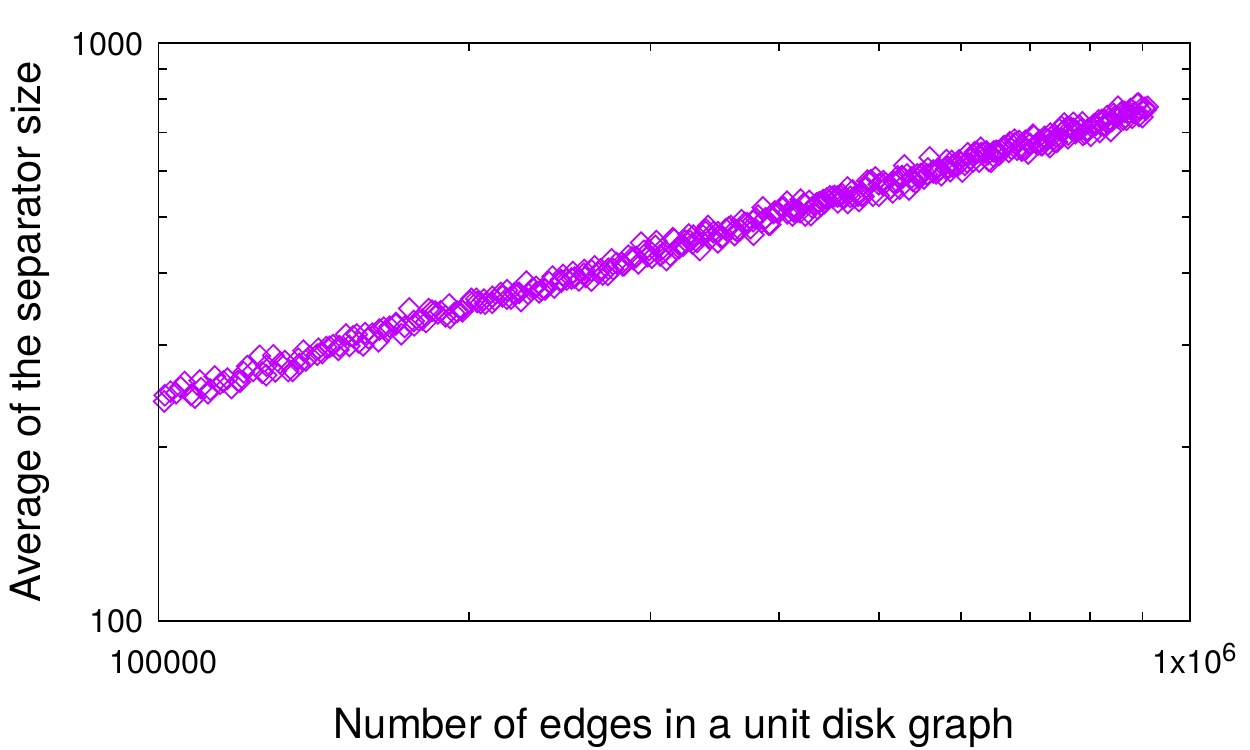}}
  \resizebox{.4\textwidth}{!}{\includegraphics{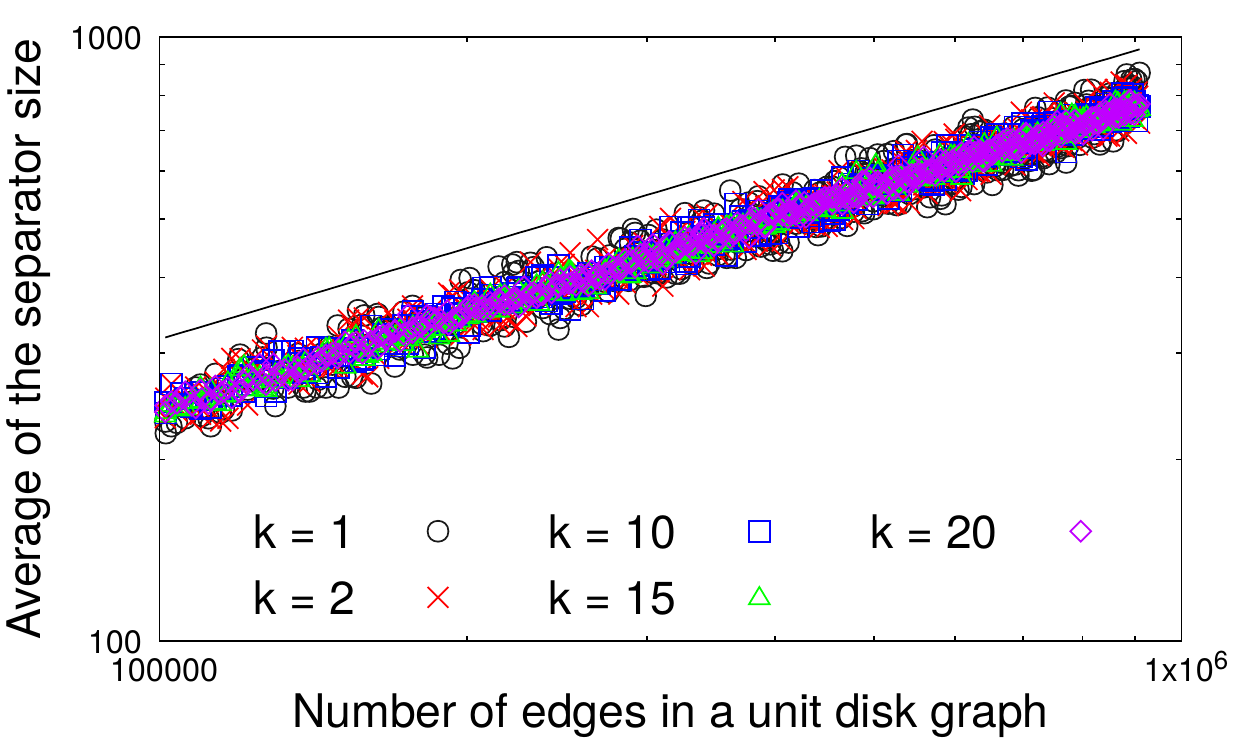}}
  \caption{
    Size of the separator obtained in our algorithm for random instances.  In all images the horizontal axis shows the number of edges in logscale and the vertical axis shows the average separator size in logscale. The images show the size of the separators when the algorithm is run $k$ times, for different values of $k$ (specifically, from top to bottom and left to right the images show $k=1,2,10,15,$ and $20$). The final image shows all plots combined in one for ease of comparison. For comparison purposes, in the last figure we have also shown with a thin black line the curve $y= a+b\cdot \sqrt{m}$ (for some constants $a,b$, line shown shifted for ease of visualization).
  }
  \label{fig:exp1-1}
\end{figure}

\begin{figure}[t]
  \centering
  \resizebox{0.8\textwidth}{!}{\includegraphics{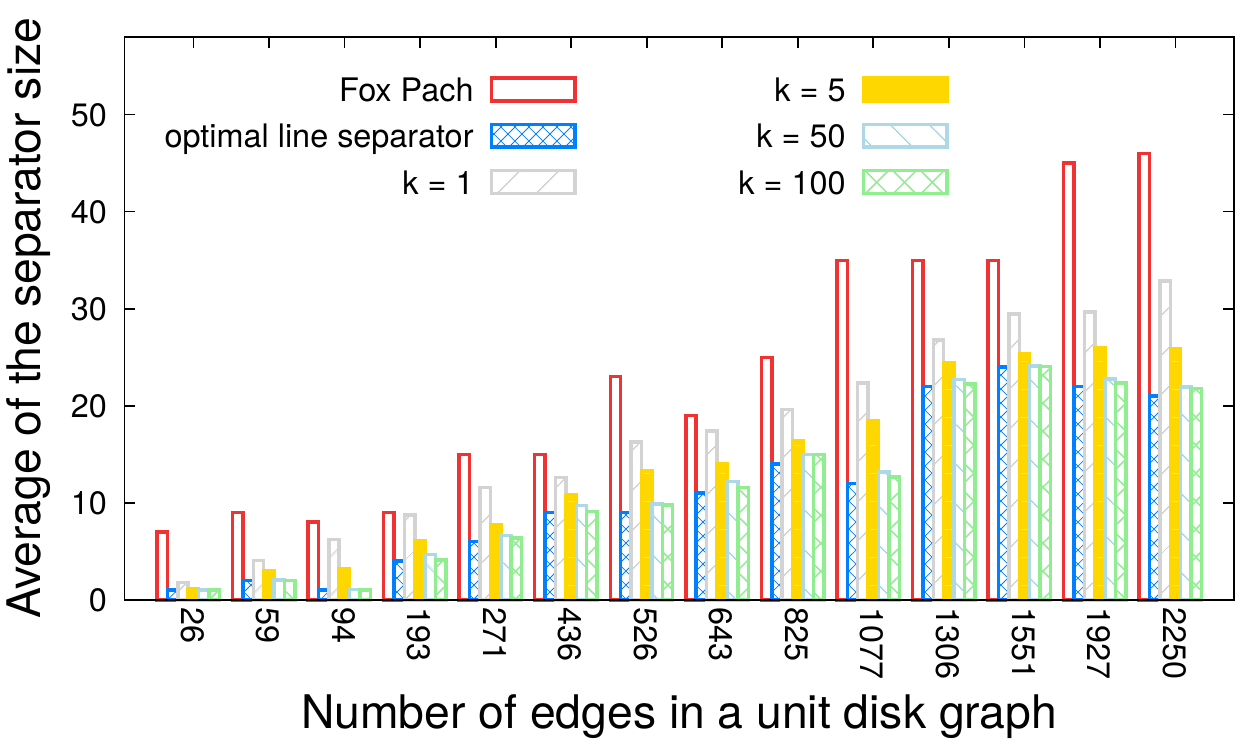}}
  \caption{
    Comparison of the separator sizes of our method with the Fox-Pach method
    for random instances.
    The horizontal axis shows the number of edges, and 
    the vertical axis shows the separator size;
    the red bar is obtained by the Fox-Pach method,
    the blue bar is an optimal line separator, 
    and the remaining four bars are obtained by our method after trying $k$ 
    random directions  ($k\in \{1,5,50,100\}$) and returning the minimum
    size separators. For all methods we repeat this process 20 times and
    display the average minimum size. 
  }
  \label{fig:exp1-2}
\end{figure}

\begin{figure}[t]
  \centering
  \resizebox{.8\textwidth}{!}{\includegraphics{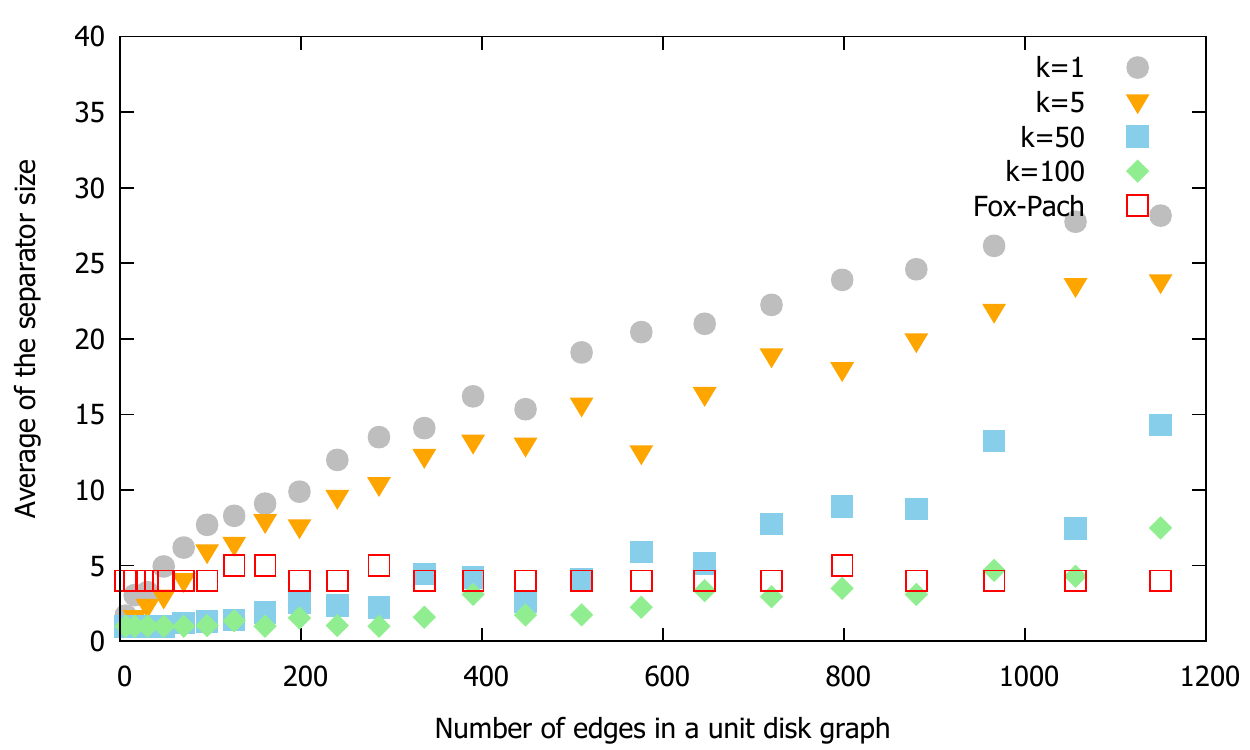}}
  \caption{
    The average separator size of our algorithm for snake instances. The horizontal axis shows the number of edges, which we control with the artificial parameter $q$ that ranges from $3$ to $199$. The vertical axis shows the average separator size when the algorithm is run $k$ times, for $k\in \{1,5, 50, 100\}$. Recall that, by construction of the snake instance, in all instances the optimal separator has size exactly one.
  }
  \label{fig:snake_plot}
\end{figure}

The computation was fast. Finding one random line separator only takes a few milliseconds. \figurename~\ref{fig:exp1-1} shows plot of the separator sizes as a function of the number of edges in the instance.
As expected, the size of the line separator increases as the number of edges increase. 
However, the experiments show something more interesting: the size of the separator seems to have an asymptotic  $\Theta( \sqrt{m})$ behavior. This further pushes our intuition that the upper bound of Theorem~\ref{thm:joint_unit} is loose. We also observe that there is high variance when $k$ is small (that is, only few lines are selected when computing the average), but this variance shrinks as $k$ grows. 

\subsection{Experiment 2: Comparison with the Method by Fox and Pach}
In the second experiment we compare the separator size of our algorithm with the one obtained with the method by
Fox and Pach.
For random instances, we fix the side length of the square to be $16$, and generate 14 random instances
with various numbers of edges (the number of disks is not fixed).
For each instance, we run our implementation of the Fox-Pach method
and our simplified algorithm described earlier.

Unlike in the first experiment, this time our algorithm picks $k$ random directions (where $k\in \{1,5,50,100\}$) and we take minimum separator size among the $k$ directions. This way, as $k$ grows to infinity, the size of the separator should get closer and closer to the real minimum. For a fixed instance and value of $k$, we run our algorithm 20 times and compute the average size of the returned separators.

We also examine the size of an optimal line separator to investigate
the limitation of the randomized line-separator approach.
Here, an \emph{optimal line separator} means the line that intersects the fewest disks, that also is a $2/3$-separator. Such a line can be found in $\tilde{O}(n^2)$ time by looking at all possible bitangents and finding the best $2/3$-separator among those that have that fixed direction, where the $\tilde{O}(\cdot)$ notation suppresses logarithmic factors. Since running time is not our main focus of interest, we instead implemented a simpler $O(n^3)$-time algorithm.

\figurename~\ref{fig:exp1-2} shows the result of the computation.
On average, our method is strictly better than the Fox-Pach method
even when $k=1$ (that is, both methods provide a $2/3$-separator, but ours intersects fewer disks).
As $k$ increases, the average separator size approaches the size of the
optimal line separator. Empirically, 50 iterations are enough to obtain a reasonably good
line separator, but even with 100 iterations we cannot obtain
the optimal separator for large instances.

We repeated the same experiment for snake instances. To generate different instances, we vary the parameter $q$ from $3$ to $199$ with an increment of $2$.

\figurename~\ref{fig:snake_plot} shows the result of this third experiment.
For random line separators, 
we observe the tendency that the average separator size decreases as
$k$ increases.
Another observation is that the random line separator rarely finds an
optimal solution. Recall that the optimal separator has size one, but the possible range of directions that achieve this tends to zero as $q$ grows.
On the other hand, 
the separator size by the Fox-Pach method is steady around four and five.
We note that our implementation of the Fox-Pach method spends more than five minutes for the
computation when the number of edges is $576$ (i.e., when $q=33$) whereas our algorithm needs miliseconds per instance.

\begin{figure}[h]
  \centering
  \centering
  \resizebox{0.3\textwidth}{!}{\includegraphics{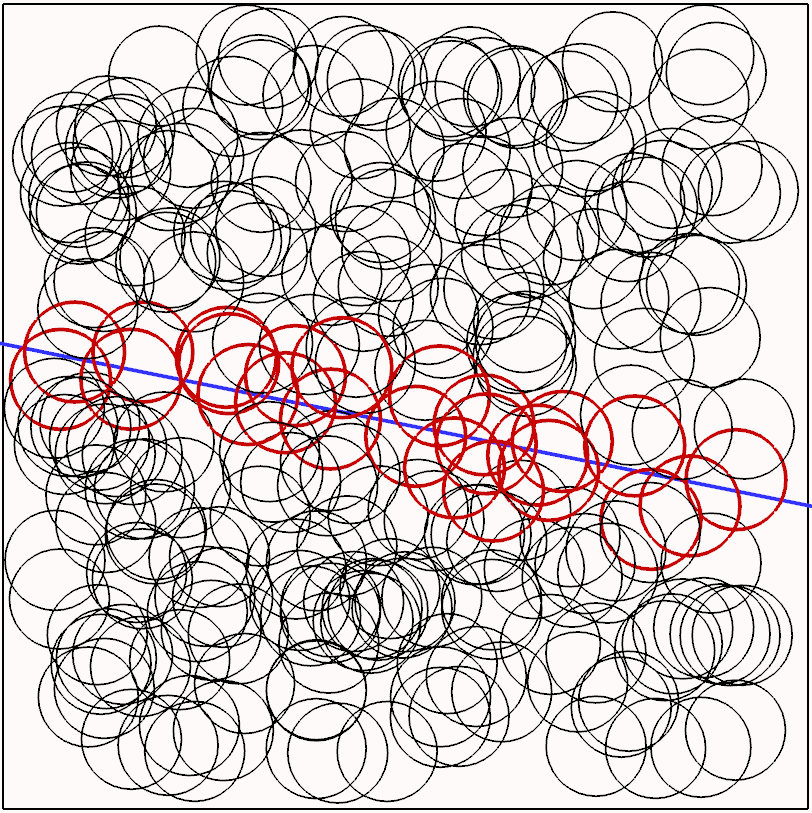}}
  \qquad
  \resizebox{0.3\textwidth}{!}{\includegraphics{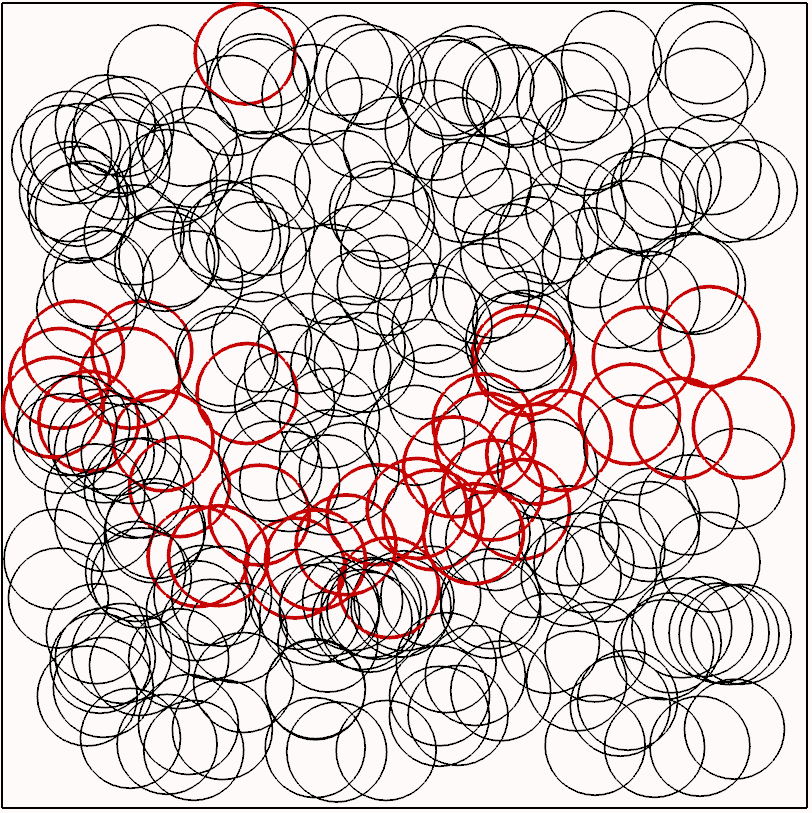}}
  \caption{
    The obtained separators when the number of edges is 1,551.
    Red disks form a $2/3$-separator.
    (Left) An optimal line separator of size 24.
    (Right) A separator obtained by the Fox-Pach method of size 35.
  }
  \label{fig:exp2-disk}
\end{figure}

Finally, \figurename~\ref{fig:exp2-disk} shows the shape of the obtained separators
for a random instance.
As the picture shows, a line separator is more geometrically appealing, whereas
the separator by the Fox-Pach method is geometrically more complex.
In particular, we note that the Fox-Pach method chooses a disk on the ``outer boundary'' that is often disconnected from the other disks. 
This is a result of additional artificial edges added to triangulate the graph.

\section{Conclusion}

The paper leaves some open questions.
The proof of Theorem~\ref{thm:joint_unit} relies on the existence 
of a centerpoint, 
and the balance of $2/3$ looks intrinsic with this approach.
It is an open problem to give a better bound on the balance
of a line separator,
possibly at the expense of a worse bound on the separator size.
This may lead to a trade-off between
the balance and the separator size for line separators.

This paper concentrated on disks, and in particular unit disks.
The method explained at the end of Section \ref{subsec:generalcase}
can be used to show a similar result to Theorem~\ref{thm:joint_unit}
for \emph{fat} convex objects of similar size \cite{DBLP:journals/algorithmica/BergSVK02}.
However, we do not know whether the method can be extended to arbitrary convex objects of similar size.
This is another open problem.

\paragraph{Acknowledgments.}

The authors thank Michael Hoffmann and Eli Fox-Epstein for 
motivating discussion on the topic.
We also thank the anonymous referees for their careful and constructive criticism that helped improve the presentation of this document.

\bibliographystyle{plain}
\bibliography{biblioSepar}

\end{document}